\documentclass[pra,twocolumn,10pt, superscriptaddress,notitlepage,showpacs]{revtex4-1}
\usepackage{amsmath,amsfonts,amssymb,amstext,amscd,amsthm,dsfont,bbm,hyperref,natbib}
\usepackage{physics}
\usepackage{color}
\usepackage{soul,xcolor}
\usepackage{gensymb}
\hypersetup{
    colorlinks = true,
    citecolor=blue,
    linkcolor = red,
    anchorcolor = red,
    citecolor = blue,
    filecolor = red,
    pagecolor = red,
    urlcolor = red,
    }
\usepackage[normalem]{ulem}
\usepackage{graphicx}
\usepackage{bbold}
\usepackage{braket}
\usepackage{placeins,comment}

\newtheorem{Theorem}{Theorem}

\newtheorem{Remark}{Remark}
\newtheorem{Example}{Example}
\newtheorem{Proposition}{Proposition}
\expandafter\let\csname equation*\endcsname\relax
\expandafter\let\csname endequation*\endcsname\relax

\usepackage{graphicx}
\usepackage{hyperref}
\usepackage{float}
\usepackage[normalem]{ulem}
\usepackage{subcaption}
\usepackage{epstopdf}
\usepackage{color}
\begin{document} %Document details
\title{Eternal non-Markovianity of qubit maps}
\author{Vinayak Jagadish}
\affiliation{Centre for Quantum Science and Technology, Chennai Institute of Technology, Chennai, 600069, India}
				\author{R. Srikanth}
\affiliation{Theoretical Sciences Division,
Poornaprajna Institute of Scientific Research (PPISR), 
Bidalur post, Devanahalli, Bengaluru 562164, India}

%\date{} 

\begin{abstract} 
		As is well known, unital Pauli maps can be eternally non-CP-divisible. In contrast, here we show that in the case of non-unital maps, eternal non-Markovianity in the non-unital part is ruled out. In the unital case, the eternal non-Markovianity can be obtained by a convex combination of two dephasing semigroups, but not all three of them. We study these results and 
        ramifications arising from them.
\end{abstract}
%\pacs{03.65.Yz,03.67.-a}
\maketitle  
\section{Introduction and Preliminaries}
Quantum systems are not isolated and are always subjected to interactions with other systems that we have no control of or little access to, known as the environment~\cite{Davies1976-ab,haroche_exploring_2006}. Since the dynamics of the environment is difficult to be controlled or monitored, their effects manifest in the dynamics of the system as ``undisclosed measurements". The dynamics of an open quantum system is described by a dynamical map $\Phi(t)$ that acts on the set of states of the system $\rho(t) =\Phi(t)\rho(0)$. These maps are completely positive (CP) and trace-preserving (TP)~\cite{Quanta77}. The coherent evolution of quantum systems is usually hampered by the unavoidable coupling of quantum systems with the external environment. Countering these detrimental effects of external noise forms the basis of designing new quantum technologies. Traditionally, noise was treated as memoryless or Markovian, but it has now been recognized that this assumption is not totally true. Quantum non-Markovianity has become an attracting theme, therefore, owing to both the fundamental aspects and its impact on developing quantum technologies. 

A classical discrete-time stochastic process $\{x_t\}$ is Markovian if and only if the probability distribution at time $t_{i+1}$ depends only on the state of affairs at time $t_i$, i.e.,
\begin{equation} \label{classdef}
\begin{split}
& p(x_t|x_{t_i},x_{t_{i-1}},\dots,x_{t_1},x_0) = \ p(x_t|x_{t_i}),\\&  \forall \thinspace i\geq1, t>t_i>t_{i-1}>\dots t_1>0, 
\end{split}
\end{equation} 
if the conditional probabilities on both sides are well defined. Thus, classical Markovianity is memoryless and is totally independent of the history of the dynamics of the system of interest.
However, quantum observables in general are non commutative, and even so for the same observable at different times. This restricts the existence of an analogous joint probability density for quantum systems. This makes the definition of Markovianity in the quantum regime nontrivial. Several definitions of quantum Markovianity have been proposed (see, e.g.,~\cite{breuer_colloquium:_2016,de_vega_dynamics_2017,li_concepts_2017,rivasreview}) and therefore the concept is not unique. However, two definitions of quantum Markovianity have been of interest:  

Decreasing distinguishability of states: For a classical Markovian process $\{x_t\}$ as in Eq.~(\ref{classdef}), the probability of correctly distinguishing between two evolving probability densities  $p_1(x_t)$ and $p_2(x_t)$ decreases monotonically in time. Breuer, Laine and Piilo (BLP) adapted this to the quantum regime identifying quantum Markovianity with the decreasing distinguishability of the states of the system in consecutive times~\cite{breuer-prl-2009}. The BLP definition of quantum Markovianity corresponds to a monotonic decreasing trace distance ${\rm tr}|\rho_1(t)-\rho_2(t)|$, for all pairs $\rho_1(t)$ and $\rho_2(t)$ and times $t\geq0$.

Existence of a completely positive dynamical map connecting time-sequential states of the system (CP divisibility): Rivas, Huelga, and Plenio (RHP)~\cite{rivas_entanglement_2010} identified quantum Markovianity with the existence of a completely positive map $\Phi_{t,t'}$ connecting the states of a quantum system at any two consecutive times $t'<t$. The divisibility of the dynamical map expresses the idea that the dynamics from initial time $t_i$ to final time $t_f$ can be represented as a concatenation of intermediate propagators. This may be expressed as follows. Given arbitrary time instants $t_f\geq t_n \geq t_{n-1} \geq \cdots \geq t_2 \geq t_1 \geq t_i,$
\begin{align}
\Phi(t_f, t_i) &= K(t_f, t_n)K(t_n,t_{n-1}) \cdots K(t_2,t_1)\Phi(t_1, t_i)
\label{cpdivdef}
\end{align}
The map is said to be Markovian (CP-divisible) if for all $t$, the intermediate propagators $K(t_m, t_{m-1})$ is CP.
CP divisibility also implies that the canonical decay rates $\gamma_i(t)$, appearing in the time-local master equation~\cite{gorini_completely_1976} given by
\begin{equation}
\begin{split}
 \mathcal{L}(t)[\rho]=&-\imath [H(t),\rho]\\&
+\sum_i \gamma_i (t)
\left(L_i(t)\rho L_i(t)^\dagger-\frac {1}{2}\{L_i(t)^\dagger L_i(t),\rho\}\right),
\label{gksl}
\end{split}
\end{equation}
are always positive. Here, $H(t)$ is the effective Hamiltonian, $L_i(t)$'s are the Lindblad jump operators. Note that $\{,\}$ denotes the anticommutator bracket.

Let us now summarize the basic definitions that we need.
\begin{itemize}
    \item Unital Maps: A dynamical map is unital if $\Phi(t)[\mathbb{1}] = \mathbb{1}$. For the corresponding time-local generator $\mathcal{L}(t)[\mathbb{1}] =0$. For non-unital maps, $\mathcal{L}(t)[\mathbb{1}] \neq 0$.

    \item Eternal non-Markovianity (ENMity): In Eq.~(\ref{gksl}), if any of the decay rates $\gamma_i(t)$ remains negative for all times $t > 0$, and furthermore $\exists_{\delta >0} \lim_{t\rightarrow\infty} \gamma_i(t)\le-\delta$, then the evolution is said to be eternally non-Markovian. The first example of ENMity was shown in~\cite{hall2010}, where one of the rates asymptotically takes on the value $-1$. Later, it was shown that this corresponds to a convex mixture of two Pauli dephasing semigroups~\cite{megier_eternal_2017}. A weak version of ENMity can be defined, wherein there exists an eternally negative rate $\gamma_i$, but no such finite bound $\delta$, i.e., $\lim_{t\rightarrow\infty} \gamma_i(t)\rightarrow 0^{-}$. 

     Quantum non-Markovianity has been demonstrated experimentally, for instance see~\cite{Bernardes2015-fm, Goswami2021-tn,Rojas-Rojas2024-ml}.  More recently, the production of non-Markovianity including eternal non-Markovianity by mixing Markovian maps has been experimentally demonstrated on an NMR quantum processor ~\cite{gulati2024}. 

    \item Quasi-Eternal non-Markovianity (Q-ENMity): If a decay rate $\gamma_i(t)$ in Eq.~(\ref{gksl}) becomes negative at some instant of time, $t = t^{\star}$ and remains negative forever (at $t=\infty)$, the map is said to be quasi-eternally non-Markovian. Here again we can define the weak version, analogous to the case of ENMity. An example of a weakly ENM generalized amplitude damping map was presented in~ \cite{shrikant2024eternal}.
\end{itemize}
In Ref. \cite{shrikant2024eternal}, the authors address the question of ENMity in the context of a restricted family of damping maps. Here we consider qubit evolution in a more general family of non-unital qubit maps, including phase covariant maps, which include damping ones as a special case. 

\begin{Proposition} 
\label{ncpnegativedcay}
A time-local master equation Eq.~(\ref{gksl}) is not completely positive if any canonical decay rate is negative at $t=0$. 
\end{Proposition}
An example is provided in Sec.~\ref{Pauli}.

\begin{Proposition}
    \label{enmlemma}
    For an eternally non-Markovian map, there is at least one decay rate in the associated time-local master equation that is zero at $t=0$.
\end{Proposition} 
\begin{proof}
By definition, an eternally non-Markovian map has at least one decay rate that is negative at all times $t>0$. By Proposition~\ref{ncpnegativedcay}, we know that a negative decay rate at $t=0$ corresponds to a NCP map. Therefore, for ENM, the decay rate which remains negative for all $t>0$ ought to be zero  at $t=0$.
\end{proof}
In this article, we discuss the conditions for obtaining eternally non-Markovian maps in the unital and non-unital qubit cases. Sec.~\ref{Pauli} deals with Pauli maps, and the discussion on the non-unital case is done in Sec.~\ref{nonunital}. Specifically, we look into the example of phase-covariant maps, which includes the damping maps as well. Finally, we summarize the discussion and pose a few open questions.
\section{Pauli Maps and Eternal non-Markovianity}
\label{Pauli}
In the context of Pauli maps that are unital, a dephasing map is never eternally non-Markovian. To show that this is not possible, without loss of generality consider the dephasing map \begin{equation}
\Phi (t)[\rho]=[1-p(t)]\rho + p(t)\sigma_3\rho \sigma_3
\end{equation} 
The associated time-local master equation can be evaluated to be
\begin{equation}
\mathcal{L}(t)[\rho]= \frac{\dot{p}(t)}{1-2p(t)}(\sigma_z\rho\sigma_z-\rho) 
\end{equation}
If this has to be eternally non-Markovian, the decay rate $\gamma(t) \equiv \frac{\dot{p}(t)}{1-2p(t)}$ has to be zero at $t=0$ and negative at all $t>0$. 
Integrating, we find $p(t) = \frac{1}{2}\left(1-e^{\big[-\frac{1}{2}\int \gamma(t) dt\big]}\right)$. For eternal non-Markovianiy, the integrand must evaluate to a negative definite number, implying that $p(t)$ will be negative for all $t>0$, which is clearly disallowed.

Consider the three Pauli dynamical semigroups
\begin{equation}
\label{paulichanndef}
\Phi_i (t)[\rho]=[1-p(t)]\rho + p(t)\sigma_i\rho \sigma_i, i= 1,2,3
\end{equation} 
with $p(t)$ being the decoherence function and $\sigma_i$ the Pauli matrices, being mixed in proportions of $x_i$ as
\begin{equation}
\label{outputmappauli}
\tilde{\Phi}(t) = \sum_{i=1}^{3} x_{i} \Phi_i (t),  \quad (x_i >0, \sum_i x_i =1).
\end{equation}
The map $\tilde{\Phi}$ satisfies the eigenvalue relations
\begin{equation}
\label{eigpauli}
\tilde{\Phi}(t)[\sigma_i] = \lambda_i (t) \sigma_i.
\end{equation}
 Let each of the dephasing maps be characterized by the same decoherence function, 
\begin{equation}
\label{decohfunc}
p(t) = \frac{1-e^{-ct}}{2}, c >0.
\end{equation}
The corresponding time-local master equation for $\tilde{\Phi}(t)$ reads as
\begin{equation}
\label{megen}
\mathcal{L}(t)[\rho] = \sum_{i=1}^{3}\gamma_{i} (t) (\sigma_i\rho\sigma_i-\rho),
\end{equation}
with the decay rates
\begin{eqnarray}
\gamma_1(t) &=& -f[x_1,p(t)]+f[x_2,p(t)]+f[x_3,p(t)] \nonumber \\
\gamma_2(t)  &=&  f[x_1,p(t)]-f[x_2,p(t)]+f[x_3,p(t)] \nonumber \\
\gamma_3(t)  &=&  f[x_1,p(t)]+f[x_2,p(t)]-f[x_3,p(t)],
\label{eq:3form}
\end{eqnarray} 
where $f[(x_i,p(t)] = \frac{1-x_i}{1-2 (1-x_i)p(t)}\frac{\dot{p}(t)}{2}\ge0$ for all $p(t) \in [0,\frac{1}{2})$.
 \begin{Theorem}
\label{Theorem1}
	 For Pauli maps, an ENM map is obtained only by the convex combination of two Pauli dephasing semigroups. Based on mixing coefficients, mixing all three maps can result only in quasi-ENMity.
\end{Theorem}
\begin{proof}
 It follows from Eq.~(\ref{eq:3form}) at most only one of the three decay rates can be negative. Also note that $\frac{df[x_i,p(t)]}{dp} = \frac{2 (1-x_i )^2}{[1-2 (1-x_i ) p(t)]^2} > 0$ for all $p(t), x_i$. Thus, in a given decay rate $\gamma_i(t)$, the negative term will produce a monotonic decrease in the rate whereas the other two terms will produce a monotonic increase. The monotony of $f[x_i,p(t)]$ tells us that if a given rate (say) turns negative at $p(t)=p_0$, then it remains negative throughout the remaining range of $p(t)$. 

Suppose that $\forall_j x_j>0$. Eq.~(\ref{eq:3form}) implies that $\lim_{t\rightarrow0} \gamma_j =  \frac{c x_j}{2} > 0$. This therefore rules out the possibility of obtaining an ENM map by mixing all three Pauli dephasing semigroups. 

Consider the two-way mixing. Without loss of generality, let $x_1=0$. Also, let us notate $x_2=a$. The decay rates are expressed in terms of the eigenvalues of the resultant map $\tilde{\Phi}(t)$ as
\begin{align}
	\gamma_1(t) &= \frac{1}{4} \bigg( \frac{\dot{\lambda}_1(t)}{\lambda_1(t)} - \frac{\dot{\lambda}_2(t)}{\lambda_2(t)} - \frac{\dot{\lambda}_3(t)}{\lambda_3(t)}  \bigg), \nonumber\\ 
	\gamma_2(t) &= \frac{1}{4} \bigg(\frac{\dot{\lambda}_2(t)}{\lambda_2(t)} -\frac{\dot{\lambda}_1(t)}{\lambda_1(t)} - \frac{\dot{\lambda}_3(t)}{\lambda_3(t)}  \bigg), \nonumber\\
	\gamma_3(t) &= \frac{1}{4} \bigg(\frac{\dot{\lambda}_3(t)}{\lambda_3(t)}  - \frac{\dot{\lambda}_1(t)}{\lambda_1(t)} - \frac{\dot{\lambda}_2(t)}{\lambda_2(t)} \bigg). 
	\label{eq:time-dep_rates}
\end{align} 
The eigenvalues of the the map $\tilde{\Phi}(t)$ are $\lambda_1=e^{-ct}, \lambda_2=1- a(1-e^{-ct}), \lambda_3=1- (1-a)(1-e^{-ct})$. The terms  $\frac{\dot{\lambda}_i(t)}{\lambda_i(t)}$ are always negative and less than or equal to $c$ in absolute value. The rates are combinations of these three terms. Also,  $\frac{\dot{\lambda}_1(t)}{\lambda_1(t)}=-c$ for $t>0$, indicating that one rate is always negative.
\end{proof}
\begin{Example}
   The celebrated example of the eternally non-Markovian map pointed out in~\cite{hall2010}, has $\gamma_1=\gamma_2=\frac{a}{2}$ and $\gamma_3(t)=-\frac{\tanh(at)}{2}( a>0)$, where precisely the one rate that shows eternal negative behavior vanishes at $t=0$. This map is obtained by the convex combination of the two Pauli dephasing semigroups, $\Phi_1(t)$ and $\Phi_2(t)$ in equal proportions.
\end{Example}
\begin{Example}
    Consider the three-way mixing of Pauli semigroups, with $x_1=0.2$, $x_2=x_3=0.4$. The map is quasi eternally non-Markovian, as the decay rate $\gamma_1(t)$ becomes negative at a time instant and then remains negative throughout.
\end{Example}

\begin{Remark} The symmetric depolarizing map, described by the
time-local master equation
 \begin{equation}
 \mathcal{L}(t) \rho= \frac{\dot{p}(t)}{1-p(t)}\sum_{i=1}^{3}(\sigma_i\rho\sigma_i-\rho).
\end{equation}
cannot be eternally non-Markovian. This follows noting that it corresponds to a mixture of three Pauli dephasing maps~\cite{jagadish_convex_2020},
\begin{equation}
\Phi (t)[\rho]=[1-\frac{3p(t)}{4}]\rho + \frac{p(t)}{4}\sum_{i=1}^{3}\sigma_i\rho \sigma_i.
\end{equation} 
Theorem \ref{Theorem1} then rules out eternal non-Markovianity.
\end{Remark}
\begin{Remark}
   Consider the affine mixing of Pauli semigroups $x_1 \Phi_1 + x_2 \Phi_2 - x_3 \Phi_3 ; x_1 +x_2 -x_3 =1$. One can easily see that the resultant dynamical map violates complete positivity, as is evident from a negative eigenvalue of the associated Choi matrix. The decay rate $\gamma_3(t)$ in the master equation, Eq.~(\ref{megen}) is negative at $t=0$.
\end{Remark}
\section{Non-unital Maps}
\label{nonunital}
We shall now look into the case of non-unital maps. To this end, let us consider the dynamical
map,
\begin{equation} 
\label{nonunitalmap}
\begin{split}
\Phi(t) [\rho] =& \frac{1}{2} \Big[{\rm tr}[\rho] \big(\mathbb{1}
+ t_3(t) \sigma_3 \big) + \lambda_1(t) {\rm tr}[\sigma_1 \rho] 
\sigma_1 +\\& \lambda_2(t) {\rm tr}[\sigma_2 \rho] \sigma_2 +
\lambda_3(t) {\rm tr}[\sigma_3 \rho] \sigma_3 \Big ].
\end{split}
\end{equation}
The associated time-local generator reads as
\begin{equation}
\label{gennonunital}
\begin{split}
\mathcal{L}(t)[\rho] &=  \gamma_+(t) \left( \sigma_+ \rho \sigma_- -
\frac{1}{2} \{\rho, \sigma_- \sigma_+ \} \right)+\\
& \gamma_-(t)
\left( \sigma_- \rho \sigma_+ - \frac{1}{2} \{\rho, \sigma_+
\sigma_- \} \right) + \gamma_3(t) \left( \sigma_3 \rho \sigma_3
- \rho \right),
\end{split}
\end{equation}
where $\sigma_{\pm} = \frac{1}{2}( \sigma_1 \pm i
\sigma_2 )$ and the decoherence rates
$\gamma_{\pm}(t)$ and $\gamma_3(t)$ are functions of
$\lambda_i(t) [i= 1,2,3]$ and $t_3(t)$ as follows.
\begin{eqnarray} 
\label{nonunitalrates}
\gamma_{\pm}(t) &=& \alpha(t)\pm\beta(t),\nonumber\\ 
\gamma_3(t)
&=& \frac{1}{4}\Bigg(
-\frac{\dot{\lambda} _1(t)}{\lambda _1(t)}-\frac{\dot{\lambda} _2(t)}{\lambda _2(t)} +\frac{\dot{\lambda} _3(t)}{\lambda _3(t)}\Bigg).
\end{eqnarray}
The expressions for $\alpha(t)$ and $\beta(t)$ are not simple, and we omit the details. However, it turns out that we do not need the explicit functional forms of $\alpha(t), \beta(t)$ for our analysis as can be seen below.

We now ask the question whether the non-unital map, Eq.~(\ref{nonunitalmap}) can be ENM in its nonunital part, i.e., the decay rates $\gamma_{\pm}(t)$ Eq.~(\ref{nonunitalrates}) being negative for all time $t>0$. We prove that this is forbidden in the sense that there can be no $\delta >0$ such that $\lim_{t\rightarrow\infty} \gamma_{\pm}(t) \le - \delta$.

\begin{Theorem}
   For the non-unital qubit map, Eq.~(\ref{nonunitalmap}) either decay rate corresponding to the non-unital part of the generator cannot be asymptotically negative.
    \label{thm:ENM}
\end{Theorem}
\begin{proof}
Let us consider the equations of motion for the Bloch vector, $\vec{R}(t)$ whose components are defined as
$r_i(t) = \mathrm{Tr}[\rho(t)\sigma_i]$.
\begin{eqnarray}
\label{blocheqns}
     \dot{r}_1(t) &=& -[\alpha(t) + \gamma_3(t)] r_1(t),\nonumber \\
    \dot{r}_2(t) &=& -[\alpha(t) + \gamma_3(t)] r_2(t),\nonumber \\
    \dot{r}_3(t) &=& -2 \alpha(t) r_3(t) - 2\beta(t).
\end{eqnarray}
From the above equation, one can see that, for $r_3(t)$ to remain bounded, $\alpha(t) \geq |\beta(t)|$. This therefore is a sufficient condition. Now, w.l.o.g., let us assume that the decay rate $\gamma_{-}(t)$ becomes  negative, at some $t = t^{*}> 0$ and stays negative throughout thereafter.  This implies that there exists a constant, $\delta>0$ such that 
\begin{equation}
\forall_{(t\geq t^{*})} \alpha(t)-\beta(t)\le -\delta.
\end{equation} 
From Eq.~(\ref{blocheqns}), we see 
\begin{equation}
\dot{r}_3(t)< -2\alpha(t)[r_3(t)+1]-2\delta.
\end{equation}
As $|r_3(t)|\leq 1$, $\dot{r}_3(t)< -2\delta$, which implies that for $t>1/2\delta, \thinspace r_3(t)<-1$, which is no longer a valid density matrix, in violation of the positivity of the map. 
\end{proof}

\begin{Example}
    For the phase covariant map [where $\lambda_1(t) =\lambda_2(t) = \lambda(t)$, in Eq.~(\ref{nonunitalmap})], we have
\begin{equation}
    \alpha(t)= -\frac{\dot{\lambda} _3(t)}{2\lambda _3(t)},\quad
\beta(t) = \frac{\dot{t}_3(t)}{2} -\frac{t_3(t)\dot{\lambda}_3(t)}{2\lambda _3(t)}
\end{equation} 
Phase-covariant qubit maps can manifest eternal non-Markovianity only in the unital part of the generator (i.e., the decay rate term $\gamma_3(t)$ corresponding to $\left( \sigma_3 \rho \sigma_3 - \rho \right)$ which was already shown in~\cite{Filippov2020-ty}. However, by Theorem~\ref{thm:ENM}, the decay rates $\gamma_{\pm}(t)$ can never be eternally negative.
\end{Example}

\begin{Remark}
    Theorem~\ref{thm:ENM} therefore says that even quasi-ENMity in the nonunital part of the qubit map is ruled out in the above strong sense. What this allows at most is quasi-eternal non-Markovianity in the non-unital part in the weak sense that $\lim_{t\rightarrow \infty} \gamma_{\pm}=0-$, i.e., asymptotically approaching zero from the negative side.
\end{Remark}

\begin{Remark}
    Following the same arguments, one can see that neither the generalized amplitude damping map [$\gamma_3(t) =0$ in the phase-covariant map] nor the amplitude damping map [where $\gamma_{-}(t) =0 =\gamma_3(t)$. in the phase-covariant map] can be quasi-eternally non-Markovian.
\end{Remark}
\section{Conclusions and Discussions}
The present work covers different aspects of eternal non-Markovianity (ENMity) in the context of unital and nonunital qubit maps. In the former case, we point out conditions under which ENMity can be produced by mixing semigroups. In the latter case, we discuss a negative instance, where ENMity is forbidden. For a class of non-unital maps, we have proved that the generator corresponding to the non-unital part can never lead to an eternally negative decay rate.  It is not obvious how to extend Theorem \ref{thm:ENM} to higher dimensions, because the proof makes use of the specific structure of the qubit affine representation. This is compounded by the fact that the parameter space becomes larger in higher dimensions. In the case of Pauli mixtures, the production of ENMity by mixing two dephasing semigroups brings forth the question of whether there exist \textit{irreducible} qudit ENM maps, namely those not obtainable as mixtures of other non-ENM maps.

Quantum memory effects can affect the properties and resource requirements in quantum information processing, in particular, lowering the threshold of quantum error correcting codes \cite{Nickerson2019-th}, and posing challenges to fault-tolerant quantum computing \cite{terhal2005fault}, necessitating specific methods to adapt quantum error correction to non-Markovian errors \cite{oreshkov2007continuous}. The eternity aspect accentuates the impact of non-Markovianity, making it important for the practical control of quantum devices.

Another interesting question is to see if a microscopic master equation representing eternal non-Markovian dynamics could be derived from first principles. In other words, can one engineer suitable system-bath interaction and initial states that could lead to ENM dynamics on the system of interest?
\acknowledgements
R.S.   acknowledges the support of Indian Science \& Engineering Research Board (SERB) grant CRG/2022/008345.
\vspace{-3 mm}
\bibliography{ENM}
\end{document}